\documentclass[a4paper]{article}
\pdfoutput=1
\usepackage{microtype}
\usepackage[colorlinks=true,urlcolor=Blue,citecolor=Green,linkcolor=BrickRed]{hyperref}
\usepackage[usenames,dvipsnames]{xcolor}
\usepackage{amsmath}
\usepackage[utf8]{inputenc}
\usepackage{algorithm}

\usepackage{algorithmicx}
\usepackage[noend]{algpseudocode}
\usepackage[noadjust]{cite}
\usepackage[OT4]{fontenc}
\usepackage{amsthm}
\usepackage{amssymb}
\usepackage{amsmath}
\usepackage{amsfonts}
\usepackage{todonotes}
\usepackage{authblk}
\usepackage{rotating}
\usepackage{xspace}
\usepackage{mathtools}
\usepackage{thm-restate}

\newtheorem{theorem}{Theorem}[section]

\newtheorem{lemma}[theorem]{Lemma}

\theoremstyle{definition}

\newenvironment{proofs}{%
  \proof}{\endproof}
\algnewcommand{\LeftComment}[1]{\Statex \(\triangleright\) #1}

\usetikzlibrary{decorations.pathreplacing}

\def\biklika #1#2#3#4#5#6#7{
  \draw[thick]  (#1, #2) rectangle ++(#3,#4);
  \node at (#1 + #3/2,#2 + #4/2) {\large  $#6_{#5}$};
  \pgfmathsetmacro\sh{0+2}
  \draw[thick]  (#1,#2+ \sh) rectangle ++(#3,#4);
 \node at (#1 + #3/2,#2 + #4/2 +\sh) {\large  $#7_{#5}$}; 
  \draw[]  (#1 + #3*1/4, #2+#4) -- ++(0,#4);
  \draw[]  (#1 + #3*3/4, #2+#4) -- ++(0,#4);
  \draw[]  (#1 + #3*1/4, #2+#4) -- ++(1,#4);
  \draw[]  (#1 + #3*3/4, #2+#4) -- ++(-1,#4);
}
\def\rest #1#2#3#4{
    \draw[thick,rounded corners=6pt]  (#1, #2) rectangle ++(#3,#4);
  \node at (#1 + #3/2,#2 + #4/2) {\large  $A_{\text{rest}}$};
  \pgfmathsetmacro\sh{#2+2}
    \draw[thick,rounded corners=6pt]  (#1, \sh) rectangle ++(#3,#4);
  \node at (#1 + #3/2,#2 + #4/2 + \sh) {\large  $B_{\text{rest}}$};  
  \draw[]  (#1 + #3*1/4, #2+#4) -- ++(0,#4);
  \draw[]  (#1 + #3*1/4, #2+#4) -- ++(1,#4);
  \draw[]  (#1 + #3*1/2, #2+#4) -- ++(0,#4);
}
\def\wielokropek #1#2#3{
  \foreach \x in {1,...,3} {
    \pgfmathsetmacro\re{\x - 1}
     \draw[fill] (#1 + \re * #3 ,#2) circle [radius=0.02];
  }
}
\def\restb #1#2#3#4{
    \draw[thick,rounded corners=6pt]  (#1, #2) rectangle ++(#3,#4);
  \node at (#1 + #3/2,#2 + #4/2) {\large  $A_{\text{rest}}$};
  \pgfmathsetmacro\sh{#2+2}
    \draw[thick,rounded corners=6pt]  (#1, \sh) rectangle ++(#3,#4);
  \node at (#1 + #3/2,#2 + #4/2 + \sh) {\large  $B_{\text{rest}}$};  
 \draw[line width=0.35mm]  (#1 + #3*1/2, #2+#4) -- ++(0,#4);
  \draw[line width=0.35mm]  (#1 + #3*1/4, #2+#4) -- ++(0,#4);
  \draw[line width=0.35mm] (#1 + #3*1/4, #2+#4) -- ++(1,#4);
}

\newcommand{\Oh}{\mathcal{O}}

\newcommand{\nodelabel}{\ell}
\newcommand{\bipnodelabel}{\ell_\text{bip}}

\newcommand{\fru}[2]{\{#1, \the\numexpr #1 + 1 \relax, \ldots, #2\}}
\newcommand{\fruu}[3]{\{#1, #2, \ldots, #3\}}

\DeclarePairedDelimiter{\ceil}{\lceil}{\rceil}
\DeclarePairedDelimiter{\floor}{\lfloor}{\rfloor}

\title{Efficient Labeling for Reachability in Digraphs}
\date{}
\author[1]{Maciej Dulęba}
\author[1]{Paweł Gawrychowski}
\author[1]{Wojciech Janczewski}
\affil[1]{Institute of Computer Science, University of Wrocław, Poland}

\begin{document}

\maketitle

\begin{abstract}  
We consider labeling nodes of a directed graph for reachability queries.
A reachability labeling scheme for such a graph assigns a binary string,
called a label, to each node.
Then, given the labels of nodes $u$ and $v$ and no other information about
the underlying graph, it should be possible to determine whether there exists a directed path from $u$ to $v$.
By a simple information theoretical argument and invoking the bound on the number of partial
orders, in any scheme some labels need to consist of at least $n/4$ bits, where $n$ is the number of nodes.
On the other hand, it is not hard to design a scheme with labels consisting of $n/2+\Oh(\log n)$ bits.
In the classical centralised setting,  Munro and Nicholson designed a data structure for reachability queries
consisting of $n^2/4+o(n^2)$ bits (which is optimal, up to the lower order term).
We extend their approach to obtain a scheme with labels consisting of $n/3+o(n)$ bits.
\end{abstract}

\section{Introduction}

A labeling scheme assigns a binary string, called a label, to each node in a graph.
Then, it should be possible to compute some function defined on subsets of nodes
using only labels of the nodes in that subset, and no other information about the whole graph.
Formally, a labeling scheme for a family of graphs consists of two parts, an encoder and a decoder.
The encoder receives a graph from the specified family and outputs the label of each node in this graph.
The label replaces the unique id of a node and allows the decoder to 
evaluate the desired function using only labels of the relevant nodes.
Therefore, such labeling schemes are often called {\em informative}~\cite{Peleg05}.
Another way of thinking about such a scheme is that we want to distribute the description
of a graph among its individual nodes.

The most important characteristic of a scheme is its size, defined as the maximum length of a label assigned to any node.
Additionally, it is desirable that the decoder is able to evaluate the function efficiently, ideally
in constant time assuming random access to all the relevant labels.
Finally, the encoder should work in polynomial time, and sometimes optimising its running time
is yet another goal.

Arguably the most basic example of a function considered in this model is adjacency: the decoder
needs to answer whether two nodes are neighbours in the graph, using only their labels.
Such a labeling scheme is closely connected to the notion of an induced universal graph for a given family of graphs,
where the induced universal graph needs to contain each graph from the family as a node-induced subgraph.
The question of the minimal size of induced universal graphs has been already studied by Moon~\cite{UniversalGraphs} several decades ago.
Recently, Alstrup, Kaplan, Thorup and Zwick \cite{AdjacencyGraphs} proved that it is possible to construct an adjacency labeling scheme
for undirected graphs with size $n/2+\Oh(1)$, which is optimal up to additive constant.
They also obtained similar tight results for directed graphs, tournaments, and bipartite graphs.
Alstrup, Dahlgaard, and Knudsen~\cite{AdjTrees} proved the optimal result for adjacency in trees, achieving labels of size $\log{n}+\Oh(1)$.
Numerous other functions were considered, both in terms of upper and lower bounds:
distance~\cite{Distance,Distance2,Distance3}, connectivity~\cite{Connectivity,Flow},
sibling or ancestor relationship~\cite{SmallDist}, nearest common ancestor in trees~\cite{NCA,NCA2}, routing~\cite{Routing} and flow~\cite{Flow}.
Often more restricted classes of graphs are analysed, most notably planar graphs~\cite{Planar1,Planar2}, bounded degree graphs~\cite{BoundedDeg}
and sparse graphs~\cite{Sparse1,Sparse2,Hubs}.
See~\cite{rotbart2016new} for a recent survey.

\paragraph{Reachability in directed graphs.}
We focus on the general class of directed graphs. Alstrup et al.~\cite{AdjacencyGraphs} considered adjacency queries
in such graphs, and designed a scheme of size $n+3$, with the obvious lower bound being $n$. The natural
next step is to consider reachability queries, in which given the labels of $u$ and $v$ the decoder should answer
if there is a directed path from $u$ to $v$. It is not hard to see that, by identifying and collapsing the strongly
connected components, it is enough to focus on directed acyclic graphs (DAGs). To extend a scheme for reachability
in DAGs to a scheme for reachability in directed graphs, we simply append $\Oh(\log n)$ bits denoting the id
of a node in its strongly connected component to the label for every node. Furthermore, we can assume that we
are given the transitive closure of a DAG, in which reachability is equivalent to adjacency.

\paragraph{Posets. }
Reachability queries in a DAG naturally correspond to comparing elements in a partially ordered set (poset).
Kleitman and Rothschild~\cite{KleitmanR70} proved the following result on the number
of posets.

\begin{theorem}[\cite{KleitmanR70}] \label{th_poset_number}
 Let $P(n)$ denote the number of posets on $n$ elements. There exists a constant $C>0$ such that
\begin{equation*}
 2^{n^2/4} \leq P(n) \leq 2^{n^2/4+ C\, n^{3/2}\log n}.
\end{equation*}
\end{theorem}

\noindent This means that supporting reachability queries in a DAG requires storing at least $n^{2}/4$ bits,
while the straightforward representation as an upper triangular matrix takes about $n^2/2$ bits.
Munro and Nicholson~\cite{MunroN16} designed a succinct data structure consisting of only $n^{2}/4+o(n^{2})$
bits for this problem.

\begin{theorem}[\cite{MunroN16}]\label{th_poset_ds}
For any poset on $n$ elements, there exists a data structure consisting of $n^2/4 + \Oh(n^2\log \log n/ \log n)$
bits supporting precedence queries in constant time.
\end{theorem}

\noindent The main idea in their approach is based on the so-called Zarankiewicz
problem, which asks about a lower bound on the number of edges in a bipartite graph guaranteeing that
there exists a balanced biclique ($K_{q,q}$) subgraph. Their construction first flattens the DAG to ensure
that there are not too many layers, namely $\Oh(\log n)$. Then, they iteratively extracts balanced bicliques with $q=\Theta(\log{n}/\log{\log{n}})$
as long as sufficiently many edges remain. The structure of a biclique allows them to encode two possible edges
with just a single bit instead of two. Finally, the remaining (not too many) edges are stored explicitly.

\paragraph{Our result. }
We translate the method of Munro and Nicholson to obtain an effective labeling scheme. This allows us
to improve on the simple upper bound of $n/2+\Oh(\log n)$ bits and obtain scheme of size $n/3+o(n)$.

\begin{restatable}{theorem}{mainresult}
\label{th_label_poset}
There exists a reachability labeling scheme for directed graphs on $n$ nodes of size $n/3+ o(n)$,
with the decoder working in constant time.
\end{restatable}

\noindent While we largely follow the approach of Munro and Nicholson, it needs to be carefully inspected and tweaked as to distribute
the stored information among the nodes. The additional ingredient is an unbalanced adjacency labeling
scheme for bipartite graphs.
Finally, we explain how to adjust the presented scheme to achieve the \emph{average} label size
of $n/4+o(n)$ at the expense of increasing the maximum label size to $n/2+o(n)$.
Other tradeoffs are also possible.
We remark that an upper bound of $n/4$ on the average label size is optimal due to Theorem~\ref{th_poset_number},
as given a labeling scheme for a DAG and all pairs of labels, the decoder can reconstruct the entire corresponding poset.

\paragraph{Overview of our approach. }
The label of every node consists of two parts.
The encoder for our scheme operates on a decomposition of the graph into antichains called layers,
with no edges between the nodes in the same layer.
First, the layers are created based on the longest-paths decomposition.
Second, we ensure that there are only $\Oh(\log n)$ layers by removing not too many edges
and merging some of the layers into super-layers. Information about the removed edges
is distributed among the first parts of the labels, each of them consisting of $o(n)$ bits.
Third, we run the following procedure that keeps removing edges from the current graph
while maintaining its decomposition into layers.
We consider the first two layers of the current graph and decompose its nodes into
balanced biclique subgraphs and the remaining nodes.
This is the key part of the construction that, roughly speaking, allows us to compress the
graph. The nodes from the bicliques are removed from the graph, and information about
their incident edges is carefully distributed among the second parts of the labels of both
the removed and the remaining nodes. After having guaranteed that the subgraph corresponding
to the remaining nodes of the first two layers is sufficiently sparse, we merge them into one layer
and repeat the reasoning.
While the idea of first flattening and then extracting bicliques is due to Munro and Nicholson~\cite{MunroN16},
we need to inspect all the ingredients and carefully balance distributing the stored information
among the labels. As a result, we end up with labels of length $n/3+o(n)$, and
with some care the decoder can be implemented to work in constant time.

\section{Preliminaries}

We consider labeling the nodes of a directed graph for reachability queries.
A labeling scheme for a family of directed graphs on $n$ nodes, denoted $\mathcal{G}_n$, consists of an encoder and a decoder.
The encoder receives a graph $G=(V,E)\in\mathcal{G}_n$ and assigns a distinct binary string (called the label) $\nodelabel_{G}(u)$
to each node $u\in V$. We will usually omit the subscript and denote the label of $u$ simply by $\nodelabel(u)$.
The decoder, given $\nodelabel(u)$ and $\nodelabel(v)$ for some $u,v\in V$,
should return if there is a directed path from $u$ to $v$ in $G$.
We stress that the decoder is not aware of $G$ and only knows that $\ell(u)$ and $\ell(v)$
are labels of two nodes from the same graph $G\in \mathcal{G}_{n}$.
We are interested in minimising the maximum length of a label, that is $\max_{G\in \mathcal{G}_n}\max_{u\in V} |\nodelabel(u)|$,
called the size of the labeling scheme.
We are also going to consider minimising the average length of a label, defined as
$\max_{G\in \mathcal{G}_n}\sum_{u\in V} |\nodelabel(u)|/n$.
When analysing the decoding time, we assume the standard Word RAM model with words of length $\Theta(\log n)$.
That is, both labels are given as arrays, with each entry storing $\Theta(\log n)$ consecutive bits of
the label, and the decoder can access any of these entries in constant time. To make our scheme
more relevant for possible applications, we insist that the decoder is uniform, that is, actually works
for any value of $n$ (otherwise the set of inputs is possibly very large but finite, and the decoding procedure
could simply access a preprocessed table, which is clearly not too practical).

$a \leadsto b$ denotes that there is a directed path (possibly with zero length) from $a$ to $b$,
and in such case we say that $a$ can reach $b$, or that $b$ is greater than $a$.

We focus on the class of directed acyclic graphs on $n$ nodes, denoted $\mathcal{DAG}_{n}$.
A labeling scheme for $\mathcal{G}_{n}$ can be obtained from our construction for $\mathcal{DAG}_{n}$
using the following lemma.

\begin{lemma}\label{lem_digraphs_to_dags}
Assume that there is a reachability labeling scheme for $\mathcal{DAG}_{n}$ of size $f(n)$
and average size $g(n)$, with the decoder working in constant time.
Then there is also a reachability labeling scheme for $\mathcal{G}_{n}$ of size $f(n)+\Oh(\log{n})$
and average size $g(n)+\Oh(\log{n})$, with the decoder working in constant time.
\end{lemma}

\begin{proof}
We explain how to obtain a labeling of the given directed graph $G=(V,E)$ by constructing
a DAG $G'=(V,E')$, using the assumed scheme to label its nodes, and prepending some extra
information to the label of every node.

$G'$ is constructed by identifying the strongly connected components (SCCs) of $G$.
Let $\{v_{1},v_{2},\ldots,v_{k}\}$ be the nodes in the same SCC. We add edges $(v_{i},v_{i+1})$, for every $i=1,2,\ldots,k-1$, to $E'$.
Then, for every edge $(u,v)\in E$ such that $u$ and $v$ belong to different SCCs consisting
of nodes $\{u_{1},u_{2},\ldots,u_{k}\}$ and $\{v_{1},v_{2},\ldots,v_{\ell}\}$, respectively, 
we add the edge $(u_{k},v_{1})$ to $E'$. It is easy to verify that, for any $u$ and $v$ belonging
to different SCCs, $u \leadsto v$ in $G$ if and only if $u \leadsto v$ in $G'$. We run our
encoder on $G'$ to obtain the label $\nodelabel_{G'}(u)$ for every $u\in V$. Then, to obtain
$\ell_{G}(u)$ we simply prepend the identifier of SCC of $u$, consisting of $\Oh(\log n)$ bits.
This allows the decoder to correctly check if $u\leadsto v$ in $G$
by first checking if they both belong to the same SCC, and if not inspecting $\nodelabel_{G'}(u)$
and $\nodelabel_{G'}(v)$.
\end{proof}

In the remaining part of the paper, we assume that the input graph $G=(V,E)$ is acyclic, and $G_c=(V,E_c)$ denotes its
transitive closure. By definition, $u \leadsto b$ in $G$ if and only if $(u,v)\in E_{c}$.
Even though the graph is directed, we will also say that such $a$ and $b$ are adjacent.

To make the decoder computationally efficient, we need the following theorem of Hagerup, Miltersen and Pagh~\cite{HagerupMP01}:

\begin{theorem}\label{th_dict}
For a given set $S \subseteq \{0,1,...,n\}$ there is a dictionary of size $\Oh(|S|)$, allowing to answer queries
$x \in S$ in constant time and constructible in time $\Oh(|S|\log{|S|})$, assuming word size $\Theta(\log{n})$.
\end{theorem}

\section{Warm-up and bipartite graphs}

We first present a very simple preliminary scheme of size $n/2+\Oh(\log n)$. We note that the underlying idea
was already implicit in the work of Moon~\cite{UniversalGraphs}.

\begin{theorem}\label{simple}
There exists a reachability labeling scheme for $\mathcal{DAG}_n$ of size $n/2+ \Oh(\log n)$,
with the decoder working in constant time.
\end{theorem}

\begin{proof}
Consider $G=(V,E)\in \mathcal{DAG}_n$ and fix an arbitrary topological numbering of its nodes $I(\cdot)$, starting from $0$.
For any $u,v\in V$, $I(u) < I(v)$ implies that there is no path from $v$ to $u$ in $G$.
The encoder for every node $u \in V$ composes $\nodelabel(u)$ out of an encoding of $I(u)$ consisting of $\log n$
bits and a $\floor{n/2}$-bit table $B_u[\cdot]$.
For $j= \fru{0}{\floor{n/2}-1}$, the encoder sets $B_u[j]=1$ iff the nodes $u$ and $I^{-1}((I(u)+j+1) \bmod n)$ are comparable,
that is, $u \leadsto I^{-1}((I(u)+j+1) \bmod n)$ or $I^{-1}((I(u)+j+1) \bmod n) \leadsto u$.
The size of this labeling scheme is $n/2+ \Oh(\log n)$.
As for the decoder, it first extracts $I(u)$ and $I(v)$ from $\ell(u)$ and $\ell(v)$.
If $I(u)=I(v)$ then $u=v$, if $I(u)> I(v)$ then there is no path from $u$ to $v$.
We are left with the case $I(u)<I(v)$.
If $I(v)-I(u)\leq \floor{n/2}$ then the bit $B_u[I(v)-I(u)-1]$ determines whether $v$ is reachable from $u$.
Otherwise the bit $B_v[n+I(u)-I(v)-1]$ gives us this information.
\end{proof}

A technical ingredient in our solution is an adjacency labeling scheme for undirected bipartite graphs,
or equivalently reachability queries for directed graphs consisting of two layers, with the edges directed from the first
layer to the second layer. The bounds from the following lemma can be also inferred from the spreading
lemma used by Alstrup, Kaplan, Thorup and Zwick~\cite{AdjacencyGraphs} by setting all $\ell_{i}$ to be
equal. In the appendix we provide a direct proof that avoids their round-robin procedure and allows us
to provide a detailed description of the decoder.

\begin{restatable}{theorem}{lembipartite}
\label{lem_bipartite_1}
Set $a,b$ and consider a family $\mathcal{K}_{a,b}$ of bipartite graphs with two layers $A$, $B$ with $a$ and $b$ nodes correspondingly.
For any natural $\alpha$, $\beta$ satisfying $a \alpha + b \beta >  ab$ there exists an adjacency labeling scheme of size $\alpha + \Oh(\log N)$
for nodes from $A$ and size $\beta + \Oh(\log N)$ for nodes from $B$, where $N=\max\{\alpha,\beta,a,b\}$,
with the decoder working in constant time.
\end{restatable}

We remark that it is not difficult to see that such a scheme exists, by applying Hall's marriage theorem
on the following auxiliary bipartite graph $G^\prime=(A^\prime,B^\prime;E^\prime)$. We set $A^\prime=A \times B$,
$B^\prime=(A \times \fru{0}{\alpha-1})\cup(B \times \fru{0}{\beta-1})$, and
connect $(u_a,u_b) \in A^\prime$ to every node of the form $(u_a,i),i \in \fru{0}{\alpha-1}$
and $(u_b,j),i \in \fru{0}{\beta-1}$, creating $\alpha+\beta$ edges in total for every node from $A^\prime$.
When $a \alpha + b \beta \geq  ab$ holds, this graph can be verified to admit a perfect matching by
Hall's marriage theorem. Such a perfect matching forms an injective function from the edges of the original
graph to the bits of the labels of desired size.
However, we do not want the decoder to store the perfect matching, or to compute it upon a query,
so we need an explicit construction.

\newpage
\section{DAG flattening} \label{sec_flatten}
We are given a transitively closed directed acyclic graph $G_c=(V,E_c)$.
Let $d[v]$ be the length of the longest directed path ending in node $v$, and
$U_i=\{v \in V: d[v]=i-1\}$.
$\mathcal{U}= (U_1,U_2,\ldots, U_k)$ is a partition of nodes of the graph into antichains called {\em layers}.
Clearly, there are no edges between the nodes in the same $U_{i}$, and by enumerating
the nodes of $U_{1},U_{2},\ldots,U_{k}$ in this order we obtain a topological sorting of $G_{c}$.

Instead of iteratively merging pairs of adjacent layers, as done by Munro and Nicholson~\cite{MunroN16},
we directly describe which layers should be merged. Let $\gamma$ be a parameter. We call a 
layer $U_{i}$ {\em thick} if $|U_{i}| > n/\gamma$, and {\em thin} otherwise. We merge intervals
of consecutive layers to create {\em super-layers}. Each thick layer forms its own separate
super-layer, whereas consecutive thin layers are glued into a single super-layer, up to the point
where its size exceeds $n/\gamma$ or when thick layer is encountered. By construction each
super-layer has one of the following types:
\begin{description}
 \item[type 1] single thick layer,
 \item[type 2] consecutive thin layers with $(n/\gamma,2n/\gamma]$ nodes in total,
 \item[type 3] consecutive thin layers with $ \leq n/\gamma$ nodes in total.
\end{description}
Furthermore, because each super-layer of type 3 is either followed by a type 1 super-layer, or is
the very last super-layer, there are $\Oh(\gamma)$ super-layers.

After having generated the set of super-layers $\mathcal{S}$, we partition the edges as follows:
\begin{equation*}
E_1= \bigcup_{S \in \mathcal{S}}(S \times S)\cap E_c, \quad E_2 = E_c\setminus E_1.
\end{equation*}
We will show how to assign labels $\nodelabel_1(\cdot)$ that allow checking if $(u,v)\in E_{1}$
given $\nodelabel_{1}(u)$ and $\nodelabel_{1}(v)$. 

\begin{lemma}\label{lem_label_flatten}
There is an assignment of labels $\nodelabel_{1}(\cdot)$ consisting of $\Oh(\log n+ n/\gamma)$ bits
that allows checking in constant time if $(u,v)\in E_{1}$, given $\nodelabel_{1}(u)$ and $\nodelabel_{1}(v)$.
\end{lemma}

\begin{proof}
Let $I(\cdot)$ be the topological ordering of $G_c$ obtained from $\mathcal{U}=(U_{1},U_{2},\ldots,U_{k})$.
We describe how to obtain $\nodelabel_{1}(u)$ for $u\in S_i \in \mathcal{S}$.
As each super-layer consists of consecutive layers, nodes of $S_i$ create an interval  $[\textsf{beg}_i,\textsf{end}_i)$
in the topological ordering.
Let $\textsf{thick}_i$ be a Boolean value denoting whether $S_i$ is type $1$ super-layer.
$\nodelabel_1(u)$ consists of numbers $I(u)$, $i$, $\textsf{beg}_i$, $\textsf{end}_i$ and bit $\textsf{thick}_i$.
If $\textsf{thick}_i=0$ the encoder appends a bit-table $C_u[\cdot]$ of length $\textsf{end}_i-\textsf{beg}_i+1 \leq 2n/\gamma+1$
to $\nodelabel_1(u)$, where $C_u[j]=1$ iff $(u,I^{-1}(\text{beg}_i + j) )\in E_1$.
If $\textsf{thick}_i=1$, then there are no edges inside this super-layer and the encoder does not append anything to
$\nodelabel_{1}(u)$.
Observe that each label consists of  $\Oh(\log n+ n/\gamma)$ bits, and the decoder
is straightforward to implement in constant time.
\end{proof}

After removing $E_{1}$ from $G_{c}$, we obtain a new graph $G'_{c}=(V,E_{2})$ consisting of only $\Oh(\gamma)$
layers, as each super-layer now becomes a layer (however, the decomposition into layers is now not based on considering
the longest paths).
We set $\gamma=\log n$, this makes the labels obtained from Lemma~\ref{lem_label_flatten} consist of only
$o(n)$ bits, while the new graph $G'_{c}$ consists of $\Oh(\log n)$ layers.
It is easy to see that $G_{c}$ is still transitively closed.

\section{Flat DAG labeling}
\label{sec:flat}
We are now given a transitively closed directed graph $G'_{c}=(V,E_2)$ with $\Oh(\log n)$ layers $U_1,U_2, \ldots, U_k$.
Our goal is to find an adjacency labeling scheme of such graphs.
As in~\cite{MunroN16} we will find and remove bicliques in consecutive layers, with the main tool being the following theorem
by Mubayi and Turan.

\begin{theorem}[\cite{MubayiT10}]\label{th_biclique}
There exists a constant $c_{\min}$, such that every undirected graph $G=(V,E)$ with $|V|\geq c_{\min}$ and
$|E| \geq 8|V|^{3/2}$ contains a biclique $K_{q,q}$, where $q=\Theta(\log|V|/\log(|V|^2/|E|))$.
This biclique can be found in $\Oh(|E|)$ time.
\end{theorem}

We remark that the above theorem will be applied only on bipartite graphs that are much denser that the
required threshold of $8|V|^{3/2}$. Additionally, any $q=\omega(1)$ would suffice for our approach.
However, this does not seem to allow for a simpler proof.

\subsection{Biclique decomposition}
We are given undirected bipartite graph $G_{bip}=(A,B;E)$, and our goal in this subsection is to partition it into bicliques.
Later, this procedure will be iteratively applied on two consecutive layers of the initial graph.
Let $G_1=G_{bip}$. As long as the current graph satisfies the conditions of Theorem~\ref{th_biclique},
we apply it to extract a biclique (removing its nodes and edges), and repeat.
See Algorithm~\ref{alg_biclique} for a detailed description of the procedure.

\begin{algorithm} 
\scriptsize
\caption{Decomposing an undirected bipartite graph into bicliques}
\label{alg_biclique}
\begin{algorithmic}[1]
\Procedure{FindBicliques}{$G_{bip}=(A,B;E)$} 
\State $G_1 \gets G_{bip}$.
\State $n \gets |A|+|B|$
\State $\mathcal{K} \gets \emptyset$  \Comment{family of bicliques}
\State $i \gets 1$ \Comment{current graph index}
\While { $G_i=(A^i,B^i;E^i)$ satisfies $w=|A^i|+|B^i|> \max(c_{min},n^{3/4})$ and $|E^i|> w^2/(\log^6 w)$} \label{alin_biclique_conditions}
\LeftComment{For large enough $n$, $|E^i|> w^2/(\log^6 w)$ implies $|E^i|> 8w^{3/2}$}
  \State $K_i=(A_i,B_i; A_i \times B_i)$ is a biclique found by applying Theorem~\ref{th_biclique}, $A_i \subseteq A^i$, $B_i\subseteq B^i$
  \State $\mathcal{K} \gets \mathcal{K} \cup \{K_i\}$
  \State $G_{i+1}=(A^i\setminus A_i, B^i \setminus B_i; E^i \setminus (A_i \times B^i  \cup A^i\times B_i))$ 
  \Comment{we remove $A_{i}$ and $B_{i}$ from $G_i$ to obtain $G_{i+1}$}
  \State $i \gets i+1$
\EndWhile
\LeftComment{we obtain the final irreducible $G_i=(A^i,B^i;E^i)$}
\State $A_{\text{rest}} \gets A^i$
\State $B_{\text{rest}} \gets B^i$ 
\EndProcedure
\end{algorithmic}
\end{algorithm}

Let $\ell$ be the final iteration of Algorithm~\ref{alg_biclique}, and $A^\prime= \bigcup_{i=1}^\ell  A_i$, $B^\prime= \bigcup_{i=1}^\ell  B_i$.
Clearly $A= A^\prime \cup A_{\text{rest}}$ and $B= B^\prime \cup B_{\text{rest}}$, see Figure~\ref{fig3}. 
The obtained decomposition admits the following properties.

\begin{figure}[h]
\begin{center}  
 \begin{tikzpicture}[
    scale=0.6,
    axis/.style={thick, ->, >=stealth'},
    important line/.style={thick},
    dashed line/.style={dashed, thin},
    pile/.style={thick, ->, >=stealth', shorten <=2pt, shorten
    >=2pt},
    every node/.style={color=black}
   ]  
\node at (-1,0.5) {\Large $A$};
\node at (-1,2.5) {\Large $B$};
\biklika{1}{0}{2}{1}{1}{A}{B}
\biklika{4}{0}{2}{1}{2}{A}{B}
\biklika{7}{0}{2}{1}{3}{A}{B}
\biklika{12}{0}{2}{1}{\ell}{A}{B}
\draw[decorate,decoration={brace,amplitude=10pt,mirror},xshift=0.4pt,yshift=-0.4pt](1,-0.2)
-- (14,-0.2) node[black,midway,yshift=-0.6cm] {$A'$};
\draw[decorate,decoration={brace,amplitude=10pt},xshift=0.4pt,yshift=0.4pt](1,3.2)
    -- (14,3.2) node[black,midway,yshift=+0.6cm] {$B'$};
\rest{15}{0}{2}{1}
\wielokropek{10}{0.5}{0.4}
\wielokropek{10}{2.5}{0.4}
\draw (2.9,1) -- (4.2,2);
\draw (15.2,1) -- (13.9,2);
\draw (9,1) -- (15.2,2);
\draw (5.9,1) -- (12.2,2);
\draw (13.8,1) -- (15.4,2); 
 \end{tikzpicture}
\caption{$G_{bip}=(A,B;E)$  partitioned into bicliques $(A_i,B_i)$ and the leftovers ($A_{\text{rest}},B_{\text{rest}}$).} 
\label{fig3}
\end{center}  
\end{figure}
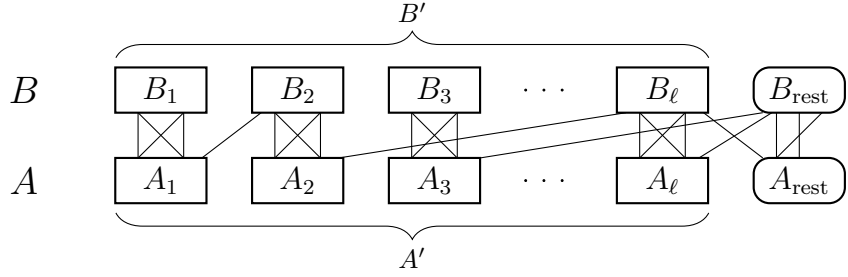

\begin{lemma}\label{lem_bipartite_rest}
 Let $V_{\mathrm{rest}}= A_{\mathrm{rest}} \cup B_{\mathrm{rest}}$ and $E_{\mathrm{rest}}=E \cap (A_{\mathrm{rest}} \times B_{\mathrm{rest}})$.
 There exists a function $N_{\mathrm{rest}}: V_{\mathrm{rest}} \rightarrow \mathcal{P}(V_{\mathrm{rest}})$ 
 such that for every edge $(u,v) \in E_{\text{rest}}$  we have $u \in N_{\mathrm{rest}}(v)$ or $v \in N_{\mathrm{rest}}(u)$.
 Also $|N_{\mathrm{rest}}(u)|= \Oh( n/ \log^3 n)$ holds  for every $u \in V_{\mathrm{rest}}$.
\end{lemma}

\begin{proof}
$G_i$ does not satisfy condition from the line~\ref{alin_biclique_conditions} of Algorithm when $A_{\text{rest}}$ and $B_{\text{rest}}$ are created.
There are two possible cases:
\begin{itemize}
 \item $|A_{\text{rest}}| + |B_{\text{rest}}|=|A^i|+ |B^i| \leq \max(c_{\text{min}},n^{3/4})$: 
 set $N_{\text{rest}}(u)= \{v \in V_{\text{rest}}: (u,v)  \in E_{\text{rest}}\}$, it satisfies all the conditions. 
 \item $|E_{\text{rest}}|< w^2/(\log^6 w) \leq n^2/(\log^6 n)$: 
 let $ V_{\text{rest}}= V_{\text{big}} \cup V_{\text{small}}$, 
 where $V_{\text{big}}$ is a set of nodes having at least $n/\log^3 n$ incident edges in the set $E_{\text{rest}}$, 
 whereas $V_{\text{small}}$ are the remaining nodes.
 For $u \in V_{\text{small}}$ set $N_{\text{rest}}(u)= \{v \in V_{\text{rest}}: (u,v) \in E_{\text{rest}}\}$.
 For $u \in V_{\text{big}}$ set $N_{\text{rest}}(u)= \{v \in V_{\text{big}}: (u,v) \in E_{\text{rest}}\}$.
 In the second case $|N_{\text{rest}}(u)|\leq |V_{\text{big}}| \leq 2|E_{\text{rest}}|/(n/\log^3 n) \leq 2n/\log^3 n$.
 Finally, consider an edge $(u,v) \in E_{\text{rest}}$. If $u \in V_{\text{small}}$ then $v \in N_{\text{rest}}(u)$.
 Otherwise $u \in V_{\text{big}}$ and $u \in N_{\text{rest}}(v)$, no matter to which $V_*$ set $v$ belongs.\qedhere
\end{itemize}
 \end{proof}
 
\begin{lemma} \label{lem_layer_size}
For every $1 \leq i \leq \ell$, $|A_i|=|B_i|=\Theta (\log n / \log \log n)$.
\end{lemma}

\begin{proof}
We have $w=|A^i|+ |B^i|> n^{3/4}$ and $|E^i|> w^2/(\log^6 w)$. Theorem~\ref{th_biclique} finds
a biclique of size $\Theta(\log w/\log(w^2/|E^i|))=\Theta (\log n / \log \log n)$.
\end{proof}

\subsection{Encoding}
We apply Algorithm~\ref{alg_biclique} iteratively to decompose the whole $G'_{c}=(V,E_2)$.
Let $s$ be the number of the current iteration.
We take the first two of the remaining layers, $U_s$ and $U_{s+1}$, treat them as an undirected bipartite graph,
and find its biclique decomposition using  Algorithm~\ref{alg_biclique}.
We obtain the set of bicliques $\mathcal{K}$, the leftovers $(A_{\text{rest}},B_{\text{rest}})$,
and the remaining layers $U_{s+2}, U_{s+3}, \ldots, U_k$.
We will soon explain how to encode information about the edges $E_{\text{inter}}$ connecting
the nodes from bicliques to other nodes in the labels $\nodelabel_{\text{inter}}^{s}(\cdot)$.
We will also explain how to encode the information about the edges $E_{\text{in}}$ between
the nodes from (possibly different) bicliques and between $A_{\text{rest}}$ and $B_{\text{rest}}$
in the labels $\nodelabel_{\text{in}}^{s}(\cdot)$.
This allows us to remove all of these edges, and also all nodes from bicliques.
We merge $A_{\text{rest}}$, $B_{\text{rest}}$ to obtain a new layer replacing $U_{s}$ and $U_{s+1}$
and repeat the procedure.
See Algorithm~\ref{alg_label} for a detailed description, and Figure~\ref{fig5} for an illustration
of a single iteration.

\begin{algorithm} 
\scriptsize
\caption{Decomposing consecutive layers of a flat DAG into bicliques} 
\label{alg_label}
 \begin{algorithmic}[1]
 \Procedure{LabelFlatDAG}{$V,E_2,s;U_s,U_{s+1},\ldots, U_{k}$}
\If{$s=k$} 
  \Return 
\EndIf
\State $E_{s,s+1} \gets E_2\cap (U_s \times U_{s+1})$
\State \textproc{FindBicliques}($U_s,U_{s+1};E_{s,s+1}$)
\Comment{we obtain the set of bicliques $\mathcal{K}$ and the leftovers $(A_{\text{rest}},B_{\text{rest}})$}
\State $V^\prime \gets A^\prime \cup B^\prime$, $V_{\text{rest}} \gets A_{\text{rest}} \cup B_{\text{rest}}$, 
$\tilde{V} \gets V \setminus (V^\prime \cup V_{\text{rest}})$
\State $E_{\text{inter}} \gets E_2\cap ((V^\prime\times \tilde{V}) \cup (A^\prime \times B_{\text{rest}}) \cup
(A_{\text{rest}} \times B^\prime))$ 
\State $E_{\text{in}} \gets E_2 \cap ((A^\prime \times B^\prime) \cup (A_{\text{rest}} \times
B_{\text{rest}}))$ 
\State Store information about $E_{\text{in}}$ and $E_{\text{inter}}$  in the labels $\nodelabel_{\text{in}}^{s}(\cdot)$, $\nodelabel_{\text{inter}}^{s}(\cdot)$ \label{alin_label}
\State $E_2\gets E_2 \setminus ( E_{\text{in}} \cup E_{\text{inter}})$ \Comment{remove edges}
\State $V \gets V \setminus V^\prime$ \Comment{remove nodes}
\State $U_{s+1} \gets A_{\text{rest}} \cup B_{\text{rest}}$ \Comment{create a new layer} \label{alin_u2}
\State \textproc{LabelFlatDAG}($V,E_2,s+1;U_{s+1}, U_{s+2}, \ldots,U_{k}$).
\EndProcedure
\end{algorithmic}
\end{algorithm}

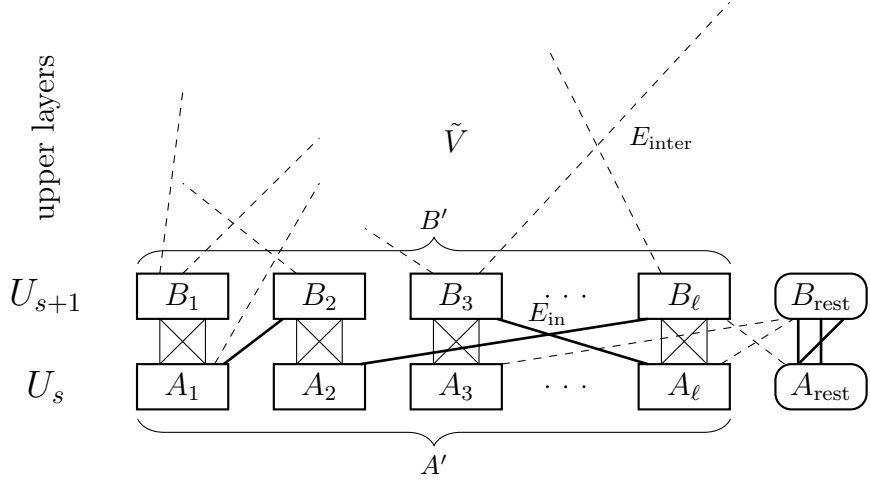
\begin{figure}[h]
\begin{center}  
 \begin{tikzpicture}[
    scale=0.6,
    axis/.style={thick, ->, >=stealth'},
    important line/.style={thick},
    dashed line/.style={dashed, thin},
    pile/.style={thick, ->, >=stealth', shorten <=2pt, shorten
    >=2pt},
    every node/.style={color=black}
   ]    
\node at (-1,0.5) {\Large $U_s$};
\node at (-1,2.5) {\Large $U_{s+1}$};
\biklika{1}{0}{2}{1}{1}{A}{B}
\biklika{4}{0}{2}{1}{2}{A}{B}
\biklika{7}{0}{2}{1}{3}{A}{B}
\biklika{12}{0}{2}{1}{\ell}{A}{B}
\draw[decorate,decoration={brace,amplitude=10pt,mirror},xshift=0.4pt,yshift=-0.4pt](1,-0.2)
-- (14,-0.2) node[black,midway,yshift=-0.6cm] {$A'$};
\draw[decorate,decoration={brace,amplitude=10pt},xshift=0.4pt,yshift=0.4pt](1,3.2)
    -- (14,3.2) node[black,midway,yshift=+0.6cm] {$B'$};
\restb{15}{0}{2}{1}
\wielokropek{10}{0.5}{0.4}
\wielokropek{10}{2.5}{0.4}
\draw[dashed] (1.5,3) -- (2,7);
\draw[dashed] (2,3) -- (5,6); 
\draw[dashed] (4.5,3) -- (2,5);
\draw[dashed] (7.5,3) -- (6,4);
\draw[dashed] (8.5,3) -- (14,9);
\draw[dashed] (12.5,3) -- (10,8); 
\draw[dashed] (2.7,1) -- (5,5); 
\draw [dashed](13.8,1) -- (15.4,2); 
\draw [dashed](15.2,1) -- (13.9,2);
\draw [dashed](9,1) -- (15.2,2);
\node[] at (12.5,6) {\normalsize $E_{\text{inter}}$};
\draw [line width=0.35mm](2.9,1) -- (4.2,2);
\draw [line width=0.35mm](8.9,2) -- (12.2,1);
\draw [line width=0.35mm](5.9,1) -- (12.2,2);
\node[] at (10,2.1) {\normalsize $E_{\text{in}}$};
\node at (8, 6) {\large $\tilde{V}$};
\node[rotate=90] at (-1, 6) {\large upper layers};
 \end{tikzpicture}
\caption{Layers $U_s$ and $U_{s+1}$ partitioned into bicliques $(A_i,B_i)$ and the leftovers 
($A_{\text{rest}}$, $B_{\text{rest}}$).} 
\label{fig5}
\end{center}  
\end{figure}

Irrespectively of the implementation of line~\ref{alin_label}, the number of iterations is $k=\Oh(\log n)$,
and the current graph $G'_{c}=(V,E_{2})$ remains transitively closed. We proceed
to explain how to implement line~\ref{alin_label}. Fix an iteration $s$ of the procedure.
Let $V^s \subseteq V$ be the set of nodes and $E_2^s \subseteq E_2$ the set of edges considered in this
iteration, with $V$ and $E_{2}$ referring to the initial graph $G'_{c}$.
The auxiliary notation ($A^\prime$, $B^\prime$, $A_{\text{rest}}$, $V^\prime$ etc.) 
refers to the sets defined in the $s$-th iteration.
We will also write just $\nodelabel_{\mathrm{in}}(\cdot)$ and $\nodelabel_{\mathrm{inter}}(\cdot)$ instead of
$\nodelabel_{\mathrm{in}}^s(\cdot)$ and $\nodelabel_{\mathrm{inter}}^s(\cdot)$.

\begin{lemma}\label{lem_label_in}
There is an assignment of labels $\nodelabel_{\mathrm{in}}(\cdot)$ consisting of
$|A^\prime|/2+\Oh(\log n)$ bits for the nodes of $V^\prime$,
$\Oh( n/ \log^2 n)$ bits for the nodes of $V_{\text{rest}}$ and $\Oh(1)$ bits for the remaining nodes
that allows checking in constant time if $(u,v) \in E_{\text{in}}$,
given $\nodelabel_{\mathrm{in}}(u)$ and $\nodelabel_{\mathrm{in}}(v)$.
\end{lemma}

\begin{proof}
For any $u\in V$, $\nodelabel_{\text{in}}(u)$ consists of the following ingredients.
First, we store an integer $\text{inf}(u)$ encoding the information whether $u$ was already
removed from the graph, or which of the sets $A^\prime$, $B^\prime$, $A_{\text{rest}}$, $B_{\text{rest}}$, $\tilde{V}$
does it belong to. Then we have two cases:
\begin{description}
\item[$u\in V^{\prime}$]: we append the label $\nodelabel_\text{bip}(u)$ obtained by applying Theorem~\ref{lem_bipartite_1}
on the bipartite graph $(A^\prime, B^\prime; E^{s}_{2}\cap (A^{\prime} \times B^{\prime}))$ with parameters $\alpha=\beta= |A^\prime|/2+1$,
\item[$u\in V_{\text{rest}}$]: we append the structure described in Theorem~\ref{th_dict} applied on the set $N_{\text{rest}}(u)$ from Lemma~\ref{lem_bipartite_rest}. 
\end{description}

Given $\nodelabel_{\mathrm{in}}(u)$ and $\nodelabel_{\mathrm{in}}(v)$, we proceed as follows.
By inspecting $\text{inf}(u)$, $\text{inf}(v)$ we can distinguish the following three options:
\begin{enumerate}
\item If $u \in A^\prime$ and $v\in B^\prime$, then using $\nodelabel_\text{bip}(u)$ and $\nodelabel_\text{bip}(v)$ we can check whether $(u,v) \in E_2 \cap (A^\prime \times B^\prime)$.
\item If $u \in A_{\text{rest}}$ and $v\in B_{\text{rest}}$, then we can check whether $u \in N_{\text{rest}}(v)$ or $v \in N_{\text{rest}}(u)$
using the dictionaries stored in both labels.
\item Otherwise $(u,v)\notin E_{\text{in}}$.
\end{enumerate}

It is straightforward to verify that the sizes of labels are as required and the check can be implemented in constant time.
\end{proof}

\begin{lemma}\label{lem_label_inter}
Let $\alpha= \ceil{|V^s|/3-|A^\prime|/2}$, $\beta=\ceil{2|A^\prime|/3}$, and $\ell$ be the number of bicliques found
in the current iteration.
There is an assignment of labels $\nodelabel_{\mathrm{inter}}(\cdot)$ consisting of
$\alpha+\Oh(\log n)$ bits for the nodes of $V^\prime$,
$\beta+ \Oh(\log n)+ \ell$ bits for the nodes of $V^s\setminus V^\prime$ and $\Oh(1)$ bits for the remaining nodes,
that allows checking in constant time if $(u,v) \in E_{\text{inter}}$, given $\nodelabel_{\mathrm{inter}}(u)$ and $\nodelabel_{\mathrm{inter}}(v)$.
\end{lemma}

\begin{proof}
We first verify that
\begin{equation*}
\alpha |A^\prime|+ \beta(|V^s|-2|A^\prime|)> |A^\prime|(|V^s|-2|A^\prime|).
\end{equation*}
Now we construct an undirected bipartite graph $\hat{G}=(\hat{A},\hat{B};\hat{E})$.
Every node of $\hat{B}$ corresponds to a node of $V^s\setminus V^\prime$. 
The definition of $\hat{A}$ is more complicated. Recall that 
$A^\prime= \bigcup_{i=1}^\ell  A_i$ and $B^\prime= \bigcup_{i=1}^\ell  B_i$.
The bicliques are balanced, so we have the natural pairing of the nodes in $A_{i}$ and $B_{i}$.
Therefore, we have a pairing of the nodes of $A^{\prime}$ and $B^{\prime}$.
Every node of $\hat{A}$ corresponds to such a pair of nodes $a \leftrightarrow b$, where
$a\in A^{\prime}$ and $b\in B^{\prime}$. Thus $|\hat{A}|=|A^\prime|=|B^\prime|$.

Observe that if for some $y \in B_i$ and $z\in V^s$ we have $(y,z) \in E_{\text{inter}}$,
then $(x,z) \in E_{\text{inter}}$ for every $x \in A_i$, by the graph being transitively closed
and $(A_{i},B_{i})$ being a biclique.
Let $\hat{a} \in \hat{A}$ correspond to $a_j \leftrightarrow b_j$, where $a_j \in A_i$, $b_j \in B_i$
(we say that $\hat{a}$ corresponds to both $a_{j}$ and $b_{j}$),
and let $\hat{b} \in \hat{B}$ correspond to $b\in V^s\setminus V^\prime$.
Whether $(\hat{a},\hat{b})\in \hat{E}$ depends on the location of $b$ in $V^{s}$ and the
edges in $E_{\text{inter}}$. Exactly one of the following cases occurs:
\begin{enumerate}\label{desc}
\item $b \in A_{\text{rest}}$, so $b$ is not adjacent to $A_i$ (in particular not to $a_j$): $(\hat{a},\hat{b}) \in \hat{E}$ iff $(b_j,b) \in E_{\text{inter}}$,   
\item $b \in B_{\text{rest}}$, so $b$ is not adjacent to $B_i$ (in particular not to $b_j$): $(\hat{a},\hat{b}) \in \hat{E}$ iff $(a_j,b) \in E_{\text{inter}}$,   
\item $b \in \tilde{V}$, $b$ is adjacent to some node of $B_{i}$, so $(a_j,b)\in E_{\text{inter}}$: $(\hat{a},\hat{b}) \in \hat{E}$ iff $(b_j,b) \in E_{\text{inter}}$,
\item $b \in \tilde{V}$, $b$ is not adjacent to any node of $B_i$, so $(b_j,b)\notin E_{\text{inter}}$: $(\hat{a},\hat{b}) \in \hat{E}$ iff $(a_j,b) \in E_{\text{inter}}$.
\end{enumerate}
We apply Theorem~\ref{lem_bipartite_1} on $\hat{G}$ with parameters $\alpha$, $\beta$ to obtain
the labels $\nodelabel_\text{bip}(\cdot)$.
For any $u\in V$, $\nodelabel_{\text{inter}}(u)$ consists of the following ingredients.
First, we store an integer $\text{inf}(u)$ encoding the information whether $u$ was already removed from the graph,
or which of the sets $A^\prime$, $B^\prime$, $A_{\text{rest}}$, $B_{\text{rest}}$, $\tilde{V}$ does it belong to.
Second, we append $\nodelabel_\text{bip}(\hat{u})$, where $\hat{u}$ corresponds to $u$ in $\hat{G}$.
Then we have two cases:
\begin{description}
\item[$u \in V^\prime$]: we append the index $i$ such that $u \in A_i \cup B_i$,
\item[$u \in V^s \setminus V^\prime$]: we append a bit-table $\text{B}_u[\cdot]$
of length $\ell$, in which $\text{B}_u[i]$ stores the information whether $u$ is adjacent to some node of $B_{i}$.
\end{description}

Given $\nodelabel_{\mathrm{inter}}(u)$ and $\nodelabel_{\mathrm{inter}}(v)$, we proceed as follows.
First we verify that $u,v \in V^s$ using $\text{inf}(u)$ and $\text{inf}(v)$, as otherwise $(u,v) \notin E_{\text{inter}}$.
Let $\hat{u}, \hat{v}$ correspond to $u$ and $v$ in $\hat{G}$.
By inspecting $\text{inf}(u)$ and $\text{inf}(v)$, we can check if $\hat{u}$ and $\hat{v}$ belong
to $\hat{A}$ or $\hat{B}$.
If $\hat{u}$, $\hat{v}$ belong both to the $\hat{A}$ or $\hat{B}$, then $(u,v) \notin E_{\text{inter}}$ and we are done.
By swapping $u$ and $v$ we can thus assume that $\hat{u} \in \hat{A}$ and $\hat{v} \in \hat{B}$.
Using $\nodelabel_\text{bip}(\hat{u})$ and $\nodelabel_\text{bip}(\hat{v})$ we can then check if $(\hat{u},\hat{v}) \in \hat{E}$.
From $\nodelabel_{\text{inter}}(u)$ we extract the index $i$ such that $u \in A_i \cup B_i$,
and by additionally inspecting $\text{inf}(u)$ we know if $u\in A_{i}$ or $u\in B_{i}$.
By inspecting $\text{inf}(v)$ we know whether $v \in A_{\text{rest}}$, $v \in B_{\text{rest}}$, or $v \in \tilde{V}$,
and by accessing the appropriate entry of $\text{B}_v[\cdot]$ we know if $v$ is adjacent to some node
of $B_{i}$. This allows us to distinguish between the four possible cases and check if $(u,v)\in E_{\text{inter}}$.
In more detail, we have the following possibilities:
\begin{enumerate}
\item $v \in A_{\text{rest}}$, if $u\in A_{i}$ then we return false, and if $u\in B_{i}$ we return $(\hat{u},\hat{v}) \in \hat{E}$,
\item $v \in B_{\text{rest}}$, if $u\in B_{i}$ then we return false, and if $u\in A_{i}$ then we return $(\hat{u},\hat{v}) \in \hat{E}$,
\item $v \in \tilde{V}$ and $v$ is adjacent to some node of $B_{i}$, if $u\in A_{i}$ we return true, and if $u\in B_{i}$
we return $(\hat{u},\hat{v}) \in \hat{E}$,
\item $v \in \tilde{V}$ and $v$ is not adjacent to any node of $B_i$, if $u\in A_{i}$ we return $(\hat{u},\hat{v}) \in \hat{E}$,
and if $u\in B_{i}$ we return false.
\end{enumerate}

It is straightforward to verify that the sizes of labels are as required and the check can be implemented in constant time.
\end{proof}

Note that in some sense the four cases from the proof of Lemma~\ref{lem_label_inter}, by the structure of the found bicliques,
allow us to store information about two possible edges ($(a_j,b), (b_j,b)$) in just a single bit.
In a similar way, Munro and Nicholson were able to obtain their centralised structure consisting of $n^2/4+o(n^2)$ bits.
Unfortunately, for a labeling scheme, when the existence of an edge from $E_{\text{inter}}$ is remembered by a node from $\hat{A}$,
one bit is used in the labels of both $a_j$ and $b_j$.
Still, only a single bit is used when the existence of an edge is stored by a node from $\hat{B}$.
This allows us to achieve a nontrivial upper bound on the total length of the label.

\begin{lemma}\label{lem_label_length}
For every $u \in V$, $\sum_{s=1}^{k} | \nodelabel_{\text{in}}^s(u)| + |\nodelabel_{\text{inter}}^s(u) | = n/3+o(n)$.
\end{lemma}

\begin{proof}
Let $i$ be the iteration in which $u$ is removed from the graph.
Recall that $V^{i}$ is the set of nodes considered in the $i$-th iteration,
and let $A^\prime_s$ denote set $A^\prime$ in the $s$-th iteration.
By Lemma~\ref{lem_label_in},
the length of $\nodelabel_{\text{in}}^s(u)$ is:
\begin{equation*}
\begin{split}
\Oh( n/ \log^2 n) &  \quad \text{for} \quad s< i \\
|A^\prime_i|/2 + \Oh(\log n)  &    \quad \text{for} \quad s= i \\
\Oh(1)         &    \quad \text{in other cases}.\\
\end{split}
\end{equation*}
This overall sums up to $o(n)+|A^\prime_i|/2$ bits, as $k=\Oh(\log n)$.
By Lemma~\ref{lem_label_inter}, the length of $\nodelabel_{\text{inter}}^s(u)$ is:
\begin{equation*}
\begin{split}
\ceil{2|A^\prime_s|/3}+ \Oh(\log n)+ \ell_s &  \quad \text{for} \quad s< i \\
\ceil{|V^i|/3-|A^\prime_i|/2} + \Oh(\log n)  &    \quad \text{for} \quad s= i \\
\Oh(1)         &    \quad \text{in other cases},\\
\end{split}
\end{equation*}
where $\ell_s$ is the number of found bicliques in the $s$-th iteration.
The sum of $2|A^\prime_s|$ over all iterations $s<i$ is equal to the number of removed nodes until the $i$-th iteration,
which is $n-|V^i|$.
The sum of $\ell_s$ is not greater than the number of found bicliques.
Because each biclique is of size $\Theta (\log n / \log \log n)$, this number is $o(n)$.
This makes the whole sum:
\begin{equation*}
o(n)+|A^\prime_i|/2 + (n-|V^i|)/3+ o(n) +|V^i|/3-|A^\prime_i|/2 + \Oh(\log n)= n/3 + o(n). \qedhere
\end{equation*}
\end{proof}

\subsection{Decoding}

We define the label $\nodelabel_2(u)$ to be the concatenation of all the labels
$\nodelabel^{s}_{\text{in}}(u)$ and $\nodelabel^{s}_{\text{inter}}(u)$ generated by Algorithm~\ref{alg_label} for $s=1,2,\ldots,k$.
Additionally, we store $\Oh(\log n)$ indices denoting where every $\nodelabel^{s}_{\text{in}}(u)$ and
$\nodelabel^{s}_{\text{inter}}(u)$ begins and ends in $\nodelabel_2(u)$. As each index needs $\Oh(\log n)$ bits,
this takes $\Oh(\log^{2}n)$ extra bits stored in the very beginning of the label, and allows us to access
any $\nodelabel^{s}_{\text{in}}(u)$ and $\nodelabel^{s}_{\text{inter}}(u)$ in constant time.
Additionally, $\nodelabel_2(u)$ stores two numbers $\text{Del}(u)$ and $\text{IU}(u)$, each in $\Oh(\log\log n)$ bits.
$\text{Del}(u)$ is the last iteration in which $u$ is present in the graph, that is, the largest $s$ such that
$u\in V^{s}$.
$\text{IU}(u)$ is the index of the initial layer of $u$ in $G'_{c}$, that is, $i$ such that $u\in U_{i}$.
By Lemma~\ref{lem_label_length}, $|\nodelabel_2(u)| = n/3+o(n)$. 

\begin{lemma}\label{lem_e2}
Given $\nodelabel_{2}(u)$ and $\nodelabel_{2}(v)$ we can check in constant time if $(u,v)\in E_{2}$. 
\end{lemma}

\begin{proof}
Every edge in $E_{2}$ ends up in exactly one of the sets $E_{\text{in}}$ or $E_{\text{inter}}$ defined in some iteration.
Note that we do not have enough time to consider all possible iterations.
Thus, we will first calculate the relevant iteration $s$, and then use $\nodelabel_{\text{in}}^s(u)$, $\nodelabel_{\text{inter}}^s(u)$,
$\nodelabel_{\text{in}}^s(v)$ and $\nodelabel_{\text{inter}}^s(v)$ to check if $(u,v)\in E_{2}$.
We will make sure that $s$ is the unique iteration such that one of the sets $E_{\text{in}}$ or $E_{\text{inter}}$
might contain $(u,v)$.

Assume that $\text{IU}(u)\leq \text{IU}(v)$, as otherwise from the topological ordering $(u,v) \notin E_2$. 
If $\text{IU}(v) \leq 2$, we take $s=1$ as the edges between the first two layers are considered only
in the first iteration.
If $\text{IU}(v)> 2$ then we have two cases:
\begin{description}
\item[$\text{Del}(u)< \text{IU}(v)-1$]: $u$ was removed in the $\text{Del}(u)$-th iteration, and before this iteration
$v$ is not in the first two layers, so we take $s=\text{Del}(u)$,
\item[$\text{Del}(u)\geq \text{IU}(v)-1$]: after the $(\text{IU}(v)-1)$-th iteration both $u$ and $v$ are in the first layer
(or not in the graph anymore) and $u$ is not in any biclique before that iteration,
so we take $s=\text{IU}(v)-1$.
\end{description}

Having identified the appropriate $s$, we use $\nodelabel_{\text{in}}^s(u)$, $\nodelabel_{\text{in}}^s(v)$ to
check if $(u,v)\in E_{in}$ and $\nodelabel_{\text{inter}}^s(u)$, $\nodelabel_{\text{inter}}^s(v)$ to check
if $(u,v)\in E_{inter}$, where $E_{in}$ and $E_{inter}$ are defined in the $s$-th iteration, in constant time.
\end{proof}

\section{Conclusions}
Lemmas~\ref{lem_label_flatten} and ~\ref{lem_e2} allow us to formulate the final theorem:
\mainresult*
\begin{proof}
By Lemma~\ref{lem_digraphs_to_dags}, it is enough to construct a reachability labeling scheme for directed acyclic
graphs on $n$ nodes of size $n/3+o(n)$ and the decoder working in constant time.
Let $G=(V,E)$ be such a DAG, and $G_{c}=(V,E_{c})$ its transitive closure. First,
we flatten $G_{c}$ to obtain a new DAG $G'_{c}=(V,E_{2})$ consisting of $\Oh(\log n)$ layers.
The set of removed edges $E_{1}$ is encoded in the labels $\nodelabel_{1}(\cdot)$ as described in
Lemma~\ref{lem_label_flatten}, using $o(n)$ bits in the label of each node and allowing
checking if $(u,v)\in E_{1}$ given the labels of $u$ and $v$, in constant time.
Next, we proceed as described in Section~\ref{sec:flat} to obtain the labels $\nodelabel_{2}(\cdot)$.
By Lemma~\ref{lem_label_length}, this uses $n/3+o(n)$ bits in the label of each node
and by Lemma~\ref{lem_e2} allows checking if $(u,v)\in E_{2}$ given the labels of
$u$ and $v$ in constant time.
Finally, the label of each node $u$ is the concatenation of $\nodelabel_{1}(u)$
and $\nodelabel_{2}(u)$, with appropriate padding as to make the length
of both parts known and allow accessing any of them in constant time.
\end{proof}

We note that the scheme can be tweaked to guarantee the optimal (up to second-order term) \emph{average} size $n/4$,
matching the centralised bound.

\begin{theorem}\label{th2}
There exists a reachability labeling scheme for directed graphs on $n$ nodes of \emph{average} size $n/4+ o(n)$,
maximum size $n/2 + o(n)$, and with the decoder working in constant time.
\end{theorem}

\begin{proofs}
To this end, we just modify Lemma~\ref{lem_label_inter}, setting $\alpha = 0$ and $\beta=|A^\prime|+1$.
Then the whole set $E_{\text{inter}}$ is remembered by the nodes in further layers, and no information about these edges is stored by the nodes from $V^\prime$.
The method from Lemma~\ref{lem_label_in} stays intact, so
the nodes from $V_{\text{rest}}$ store $o(n)$ bits and the nodes from $V^\prime$ store $|A^\prime|/2+\Oh(\log n)$ bits.
After that change, take any node $u$ and assume it is removed in the $i$-th iteration.
Then, $\nodelabel_{2}(u)$ uses one bit for every two nodes removed in the previous iterations
and one bit for every four nodes removed in the $i$-th iteration.
More precisely, recall that $A^\prime_s$ denotes the size of set $A^\prime$ in $s$-th iteration of the Algorithm~\ref{alg_label},
and let $V_{prev}$ be the set of nodes erased from the graph before iteration $i$.
Then, the label of $u$ consists of the following elements:

\begin{itemize}
\item Label $\nodelabel_1(u)$ from Lemma~\ref{lem_label_flatten}, which has length $o(n)$.
\item Labels $\nodelabel_{\text{in}}^s(u)$, with total size of $A^\prime_i/2 + o(n)$ bits as in the previous scheme.
\item Labels $\nodelabel_{\text{inter}}^s(u)$. They have lengths $|A^\prime_s|+ \Oh(\log n)+ \ell_s$ for iterations $s < i$
and $\Oh(\log n)$ for the other iterations, so the sum of their sizes is $|V_{prev}|/2+o(n)$.
\item Small additional information, that is indices denoting beginning of each sublabel and numbers $\text{Del}(u)$ and $\text{IU}(u)$.
\end{itemize}

Let us number the nodes in order of being erased from the graph,
and say nodes from the $A^\prime_i$ erased in iteration $i$ received numbers in $[a_i,b_i]$.
Then length of the label for node $u$ is $a_i/2+(b_i-a_i)/4+o(n)$.
It is easy to verify that the sum of the lengths of all the labels is at most $n^2/4+o(n)$.
This is paid for with unbalanced labels, as after the described change to $\nodelabel_{\text{inter}}^s(\cdot)$
maximum size is bounded by $n/2+o(n)$ (with the nodes from further layers having longer labels than the nodes from the previous layers).
\end{proofs}

By improving on the simple upper bound of $n/2+\Oh(\log n)$,
our result brings us closer to resolving the natural question of the space complexity of reachability labeling for directed graphs.
The only lower bound on the worst-case (and also average) size of a label in such a scheme is $n/4$, following from the result on the number of posets,
and our scheme achieves an upper bound of $n/3+o(n)$.
We remark that it does not seem possible to decrease the upper bound 
achieved by our scheme by simply tweaking the parameters, so new ideas are required.

\nocite{CliquePart}
\bibliographystyle{plain}
\bibliography{Efficient_Labeling_for_Reachability_in_Digraphs}

\appendix
\newpage

\section{Labels for bipartite graphs in constant time}
\label{sec:bipartite_proof}
\lembipartite*
\begin{proof}
Given a graph $G=(A,B;E) \in \mathcal{K}_{a,b}$, the encoder $E_\text{bip}$ first
assigns numbers $\fru{0}{a-1}$ to the nodes of $A$ and numbers $\fruu{a}{a+1}{a+b-1}$ to the nodes from $B$.
Call this assignment $I: A\cup B \rightarrow \mathcal{N}$. From now on we identify the nodes with their numbers.
The label $\bipnodelabel(u)$ of a node $u$ consists of the assigned number $I(u)$, parameters $a$, $b$, $\alpha$, $\beta$ 
($\Oh(\log N)$ bits in total) and a bit table $T_u[\cdot]$. 
If $u \in A$, the encoder sets 
\begin{equation*}
T_{u}[i]= 1 \iff (u, a + (\ceil{bu/a} +i )\bmod b) \in E \text{, for } i = \fru{0}{\alpha-1}.
\end{equation*}
If $u \in B$, the encoder sets
\begin{equation*}
T_u[j]= 1 \iff ((\ceil{a(u-a)/b} +j) \bmod a, u) \in E \text{, for } j= \fru{0}{\beta-1}.
\end{equation*}

In total labels have size $\alpha+ \Oh(\log N)$ for nodes from $A$ and $\beta + \Oh(\log N)$ for nodes from $B$.
Now we describe the decoder.
Let $u$, $v\in A \cup B$. Using $\bipnodelabel(u)$ and $\bipnodelabel(v)$, the decoder has to determine whether $(u,v) \in E$.
First, it can check whether both nodes belong to the same layer (based on $I(u)$, $I(v)$, and value $a$).
Assume that the nodes are in different layers (otherwise they are not adjacent)
and $u\in A, v\in B$ (by swapping the nodes if necessary).
Let $i_a=I(u)$, $i_b=I(v)-a$. We have $i_a \in \fru{0}{a-1}$, $i_b \in \fru{0}{b-1}$.
Let
\begin{equation*}
i=(i_b-\ceil{bi_a/a}) \bmod b,  \quad j=(i_a-\ceil{ai_b/b}) \bmod a.
\end{equation*}
If $i \in \fru{0}{\alpha-1}$, then $T_u[i]=1 \iff (u,v) \in E$.
If $j \in \fru{0}{\beta-1}$, then $T_v[j]=1 \iff (u,v) \in E$.
In both cases, the decoder can look at the right bit of the table and answer the question $(u,v) \in E$ in constant time.
So it is enough to show that for every $i_a$, $i_b$ at least one of the above holds.
When $\alpha \geq b$ or $\beta \geq a$ thesis is trivially satisfied for all $i_a$, $i_b$.
Otherwise 
\begin{equation*}
 i=\left(i_b-\left\lceil \frac{bi_a}{a} \right\rceil\right) \bmod b = \left\lfloor \frac{ai_b-bi_a}{a} \right\rfloor \bmod b,
 \end{equation*}
\begin{equation*}
 j=\left(i_a-\left\lceil \frac{ai_b}{b} \right\rceil\right) \bmod a = \left\lfloor\frac{bi_a-ai_b}{b}\right\rfloor \bmod a.
\end{equation*}
Let $w=ai_b-bi_a$. From the constraints on $i_a$, $i_b$:
\begin{equation*}
 -b(a-1) \leq w\leq a(b-1).
\end{equation*}
If $w=0$, then $i=j=0$ and we are done. 
Suppose that $w>0$, the opposite case is similar. We have 
\begin{equation*}
  i < \alpha \iff \left\lfloor\frac{w}{a}\right\rfloor \bmod b < \alpha \iff \left\lfloor\frac{w}{a} \right\rfloor < \alpha \iff  w< a\alpha,
\end{equation*}
and
\begin{equation*}
\begin{split}
 j < \beta \iff \left\lfloor\frac{-w}{b}\right\rfloor \bmod a < \beta \iff \left\lfloor\frac{ab-w}{b} \right\rfloor < \beta \\
   \iff  ab-w <b\beta \iff ab-b\beta <w.
\end{split}
\end{equation*}
From the assumption $ab-b\beta<a\alpha$, thus at least one of the above inequalities is satisfied.
\end{proof}

\end{document}